\documentclass{article}
\usepackage{spconf,amsmath,graphicx}
\usepackage{bm}
\usepackage{amsmath}
\usepackage{amssymb}
\usepackage{amsfonts}
\usepackage{amsthm}
\newtheorem{claim}{Claim}

\title{Deficient basis estimation of noise spatial covariance matrix for rank-constrained spatial covariance matrix estimation method in blind speech extraction}
\name{\shortstack{Yuto Kondo$^1$,  Yuki Kubo$^1$, Norihiro Takamune$^1$, Daichi Kitamura$^2$, Hiroshi Saruwatari$^1$ \thanks{This work was supported by Japan-New Zealand Research Cooperative Program between JSPS and RSNZ, Grant number JPJSBP120201002, and JSPS KAKENHI Grant Numbers 19K20306, 19H01116, and 19H04131.}}}

\address{$^1$The University of Tokyo,
Tokyo, Japan\\
	$^2$National Institute of Technology, Kagawa College,
Kagawa, Japan}
\begin{document}
\ninept
\newcommand{\rbm}[1]{\bm{\mathrm{#1}}}
\newcommand{\tr}{\mathrm{tr}}
\newcommand{\Htenchi}{\mathsf{H}}
\newcommand{\tenchi}{\mathsf{T}}
\newcommand{\xij}{\bm{x}_{ij}}
\newcommand{\thp}{\Theta_{{\rm p}}}
\newcommand{\Tthp}{\tilde{\Theta}_{{\rm p}}}
\newcommand{\thc}{\Theta_{{\rm c}}}
\newcommand{\Tthc}{\tilde{\Theta}_{{\rm c}}}
\newcommand{\cijt}{\bm{z}_{ij}^{({\rm t})}}
\newcommand{\cijn}{\bm{z}_{ij}^{({\rm n})}}
\newcommand{\sijt}{s_{ij}^{({\rm t})}}
\newcommand{\rijt}{r_{ij}^{({\rm t})}}
\newcommand{\rijn}{r_{ij}^{({\rm n})}}
\newcommand{\Trijt}{\tilde{r}_{ij}^{({\rm t})}}
\newcommand{\Trijn}{\tilde{r}_{ij}^{({\rm n})}}
\newcommand{\ait}{\bm{a}_{i}^{({\rm t})}}
\newcommand{\Rinilrma}{\rbm{R}_{i}'^{({\rm n})}}
\newcommand{\Rin}{\rbm{R}_{i}^{({\rm n})}}
\newcommand{\TRin}{\tilde{\rbm{R}}_{i}^{({\rm n})}}
\newcommand{\bi}{\bm{b}_{i}}
\newcommand{\vecb}{\bm{b}}
\newcommand{\biH}{\bm{b}_{i}^{\Htenchi}}
\newcommand{\lambdai}{\lambda_{i}}
\newcommand{\Tlambdai}{\tilde{\lambda}_{i}}
\newcommand{\win}[1]{\bm{w}_{i,#1}}
\newcommand{\Wi}{\rbm{W}_{i}}
\newcommand{\II}{\rbm{I}}
\newcommand{\TRinjx}{\tilde{\rbm{R}}_{ij}^{({\rm x})}}
\newcommand{\Hy}{\hat{\bm{y}}_{ij}^{({\rm n})}}
\newcommand{\Hrijt}{\hat{r}_{ij}^{({\rm t})}}
\newcommand{\Hrijn}{\hat{\rbm{R}}_{ij}^{({\rm n})}}
\newcommand{\ui}{\bm{u}_{i}}
\newcommand{\HS}{\hat{\rbm{T}}_{i}}
\newcommand{\ci}{\bm{c}_{i}}
\newcommand{\Lci}{\bm{c}_{i}^{*}}
\newcommand{\ciH}{\bm{c}_{i}^{\Htenchi}}
\newcommand{\Tci}{\tilde{\bm{c}}_{i}}
\newcommand{\TciH}{\tilde{\bm{c}}_{i}^{\Htenchi}}
\newcommand{\zi}{\bm{z}_{i}}
\newcommand{\Hcijt}{\hat{\bm{c}}_{ij}^{({\rm t})}}
\newcommand{\drawfig}[4]{ 
  \begin{figure}[#1]
  \centering \vspace{-0mm}
  \includegraphics[width=#2,clip]{#3.pdf} \vspace{-2mm}
  \caption{#4} \vspace{-4mm}
  \label{fig:#3}
  \end{figure}
}

\maketitle
\begin{abstract}
Rank-constrained spatial covariance matrix estimation (RCSCME) is a state-of-the-art blind speech extraction method applied to cases where one directional target speech and diffuse noise are mixed. In this paper, we proposed a new algorithmic extension of RCSCME. RCSCME complements a deficient one rank of the diffuse noise spatial covariance matrix, which cannot be estimated via preprocessing such as independent low-rank matrix analysis, and estimates the source model parameters simultaneously. In the conventional RCSCME, a direction of the deficient basis is fixed in advance and only the scale is estimated; however, the candidate of this deficient basis is not unique in general. In the proposed RCSCM model, the deficient basis itself can be accurately estimated as a vector variable by solving a vector optimization problem. Also, we derive new update rules based on the EM algorithm. We confirm that the proposed method outperforms conventional methods under several noise conditions.
\end{abstract}
%

\section{Introduction}
Blind speech extraction (BSE) is a technique for extracting a target speech signal from observed noisy mixture signals without any prior information, e.g., spatial locations of speech and noise sources and microphones. BSE can be interpreted as a special case of blind source separation (BSS) \cite{BSSinitial}; BSS aims to separate not only the target source but also the other sources. In this paper, we focus on the BSE problem for the observed noisy mixture that includes one directional target speech and diffuse background noise. Such BSE can be utilized for many applications including automatic speech recognition and hearing aid systems \cite{2009zanryu,BSSapplication2}.

In a determined or
overdetermined situation (number of microphones $\geq$ number of
sources), independent vector analysis \cite{intro:VBA1,intro:VBA2} and independent low-rank matrix analysis (ILRMA) \cite{intro:ILRMA,ILRMA_A,ILRMA_B,ILRMA_C} provide better BSS performance. These methods assume that the frequency-wise acoustic
path of each source can be modeled by a single time-invariant
spatial basis, which is the so-called \textit{steering vector}. In this model, the rank of a spatial covariance matrix (SCM) \cite{intro:FSCM} becomes unity in all frequencies. Thus, hereafter, we call these BSS techniques \textit{rank-1 methods}. Under diffuse noise conditions, rank-1 methods cannot separate directional sources in principle \cite{2009zanryu}, and the estimated directional sources are always contaminated with a diffuse noise component  remaining in the same direction. This is because the sources modeled by the rank-1 SCM (i.e., the steering vectors) cannot represent such spatially distributed noise components.

In contrast to rank-1 methods, multichannel nonnegative matrix factorization (MNMF) \cite{MNMFOzerov,intro:MNMF} can represent the spatially spread sources and diffuse noise because MNMF utilizes a full-rank SCM for each source. However, the estimation of the full-rank SCM has a huge computational cost and lacks
robustness against the parameter initialization \cite{intro:ILRMA}. FastMNMF \cite{intro:fMNMF1,intro:fMNMF2} is a model-constrained version of original MNMF and achieves an efficient SCM optimization with lower computational cost compared with that in MNMF, although its performance still depends on the initial values for the parameters.

To achieve fast and stable BSE under the diffuse noise condition, rank-constrained SCM estimation (RCSCME) \cite{INTRO:KUBO,LAST:KUBO} was proposed, where the mixture of one directional speech and diffuse background noise is assumed. Fig. \ref{fig:flow} illustrates a process flow of RCSCME. In RCSCME, a rank-1 method such as ILRMA is utilized as a preprocess for BSE. From the rank-1 method, $M$ estimated signals are obtained \cite{2009zanryu}; one includes target speech components contaminated with diffuse noise in the same direction and the other $M-1$ estimates consist of diffuse noise components in various directions, where $M$ is the number of microphones. Then, the rank-($M-1$) SCM of the diffuse noise are calculated from the $M-1$ noise estimates. RCSCME estimates both the deficient rank of the noise SCM, which is used to compose the full-rank SCM for diffuse noise, and the source model parameters based on the expectation-maximization (EM) algorithm. The estimated target speech signal can be obtained via multichannel Wiener filtering using the full-rank SCM of diffuse noise. In \cite{INTRO:KUBO}, it was confirmed that RCSCME can outperform ILRMA, MNMF, and FastMNMF in terms of the speech extraction performance.
\begin{figure*}[t]
 \begin{center}
  \includegraphics[scale=0.29]{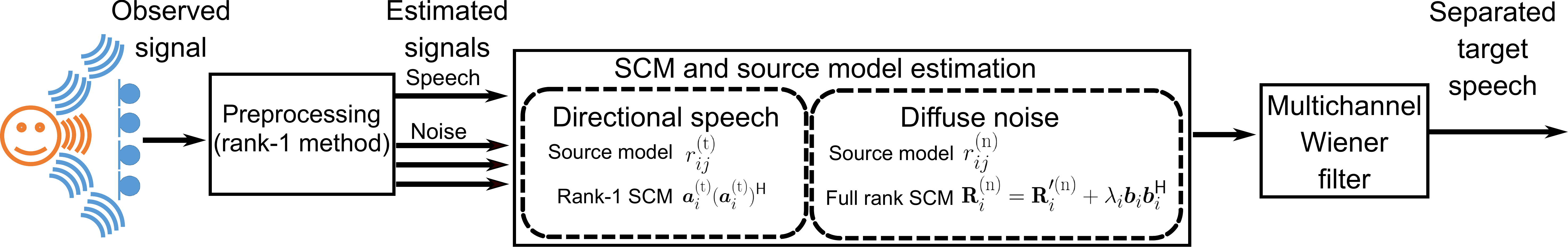}
  \vspace{0mm}
  \caption{Process flow of RCSCME.}
  \vspace{-8mm}
  \label{fig:flow}
 \end{center}
\end{figure*}

In this paper, we further improve the BSE performance by generalizing full-rank SCM estimation in conventional RCSCME. In the conventional method, the direction of the deficient basis in the rank-$(M-1)$ noise SCM is fixed to its eigenvector of the zero eigenvalue and only the scale of this basis is parameterized. 
However, the candidate of this deficient basis is not unique because any vector that is not included in the space spanned by column vectors of the rank-($M-1$) SCM can be used. In the proposed method, to estimate the optimal full-rank noise SCM, we parameterize the deficient vector itself and derive new update rules based on EM algorithm. Regarding its relation to prior works, the proposed RCSCME is interpreted as the world's first {\it spatial} model extension of the conventional RCSCME; this extension had been considered as a difficult {\it vector} optimization problem but this paper can successfully give the effective solution with a new mathematical claim.  Also, since the advantage in accurate SCM estimation under the fixed rank-($M-1$) noise SCM is still valid, the proposed RCSCME always outperforms the conventional full-rank SCM methods like MNMF and FastMNMF.

\section{Conventional RCSCME \cite{INTRO:KUBO}}
\subsection{Generative model}
In this section, we explain the generative model in conventional RCSCME. Let $\bm{x}_{ij}=(x_{ij,1}, \cdots, x_{ij,m}, \cdots, x_{ij,M})^{\tenchi}$ be the observed $M$-channel vector obtained by short-time Fourier transform (STFT)
, where
$~^{\tenchi}$ denotes the transpose, and $i=1,2,\dots, I$, $j=1,2,\dots, J$, and , $m=1,2,\dots, M$ are
the indices of frequency bins, time frames, and channels (microphones), respectively. The observed signal is a mixture of the directional target speech and diffuse noise as
\begin{equation}
\xij = \cijt + \cijn,
\end{equation}
where $\cijt\in\mathbb{C}^M$ and $\cijn\in\mathbb{C}^M$ are the source images of the target source and the diffuse noise, respectively.

The directional target source $\cijt$ can be modeled as
\begin{align}
\label{eq:cas}
\cijt     &=    \ait\sijt,\\
\sijt|\rijt &\sim \mathcal{N}_{c}(0,\rijt),
\end{align}
where $\ait\in \mathbb{C}^M$, $\sijt$, $\rijt>0$ are the steering vector, dry source component, and the time-frequency variance of the directional target source, respectively, and $\mathcal{N}_c(0, r)$ is the zero-mean circularly symmetric complex Gaussian distribution with a variance $r$. Note that $M$ steering vectors, $\bm{a}_{i,1}, \bm{a}_{i,2}, \dots, \bm{a}_{i,M}$, are estimated in advance via the rank-1 method, and we define the $n_\mathrm{t}$th steering vector $\bm{a}_{i,n_\mathrm{t}}$ corresponds to $\bm{a}_i^{\mathrm{(t)}}$. This channel selection of the target source can be achieved based on the kurtosis values of each estimated signal \cite{INTRO:KUBO}. The model (\ref{eq:cas}) assumes that the source image of directional target source has time-invariant acoustic paths represented by the single spatial basis, steering vector $\ait$, and the rank of SCM for $\cijt$ becomes unity.
In addition, since the power spectrogram of speech signals has sparsity property, we assume $\rijt\sim\mathcal{IG}(\alpha,\beta)$ as a prior distribution, where $\mathcal{IG}(\alpha,\beta)$ is the inverse gamma distribution with the shape parameter $\alpha >0$ and the scale parameter $\beta >0$.

In contrast to (\ref{eq:cas}), the source image of diffuse noise, $\cijn$, should have a full-rank SCM. The generative model of such sources is defined as \cite{intro:FSCM}
\begin{equation}
\cijn \sim \mathcal{N}_{c}^{\mathrm{(multi)}}(\bm{0},\rijn\Rin),
\end{equation}
where $\mathcal{N}_c^{\mathrm{(multi)}} (\bm{0}, \bm{R})$ is the zero-mean multivariate circularly symmetric complex Gaussian distribution with a covariance matrix $\bm{R}$, and $\rijn>0$ and $\Rin\in\mathbb{C}^{M\times M}$ are the time-frequency variance and the time-invariant SCM for diffuse noise, respectively. The full-rank SCM of diffuse noise is modeled as
\begin{align}
\Rin &= \Rinilrma + \lambdai\bi\biH,\label{eq:byB}
\end{align}
where $\Rinilrma\in\mathbb{C}^{M\times M}$ is the rank-$(M-1)$ SCM estimated by the rank-1 method  in advance and is defined as 
\begin{align}
\Rinilrma &= \frac{1}{J}\sum_{j}\Hy(\Hy)^{\Htenchi},
\label{eq:helloRilrma} \\
\Hy &= \Wi^{-1}(\bm{w}_{i,1}^{\Htenchi}\xij,\ldots,\bm{w}_{i,n_t-1}^{\Htenchi}\xij,\nonumber\\
&\phantom{=}\phantom{=}\phantom{=}\phantom{=}  0,\bm{w}_{i,n_t+1}^{\Htenchi}\xij,\ldots,\bm{w}_{i,M}^{\Htenchi}\xij)^{\tenchi},\label{eq:rit}
\end{align}
$\bm{w}_{i,m}$ is the demixing filter estimated by the rank-1 method such as ILRMA, $\Hy$ is the sum of diffuse noise components whose scales are fixed by the back projection technique \cite{intro:projectionback}, and $^\Htenchi$ denotes the Hermitian transpose. Also, $\bi\in\mathbb{C}^M$ is the deficient basis for the rank-$(M-1)$ SCM $\Rinilrma$, which makes the SCM $\Rin$ full-rank. Note that $\bi$ is defined as a unit vector and its scale is defined by $\lambdai>0$. In conventional RCSCME, the direction of $\bi$ is fixed to the eigenvector of the zero eigenvalue in $\Rinilrma$ and only the scale $\lambdai$ is estimated.

\subsection{Update rules based on EM algorithm}
The parameters $\thc=\{\rijt,\rijn,\lambdai\}$ can be optimized by
maximum a posteriori (MAP) estimation based on the EM algorithm with the latent variables $\sijt$ and $\cijn$. Details of the derivation are described in \cite{INTRO:KUBO}. The update rules are obtained as follows: In the E-step,
\begin{align}
\TRin &= \Rinilrma + \Tlambdai\bi\biH,\\
\TRinjx &= \Trijt\ait(\ait)^{\Htenchi}+\Trijn\TRin,\\
\Hrijt &= \Trijt-(\Trijt)^2(\ait)^{\Htenchi}(\TRinjx)^{-1}\ait\nonumber\\
&\phantom{=}\mbox{}+|\Trijt\xij^{\Htenchi}(\TRinjx)^{-1}\ait|^2,\\
\Hrijn &= \Trijn\TRin-(\Trijn)^2\TRin(\TRinjx)^{-1}\TRin\nonumber\\
&\phantom{=}\mbox{}+(\Trijn)^2\TRin(\TRinjx)^{-1}\xij\xij^{\Htenchi}(\TRinjx)^{-1}\TRin  ,
\end{align}
and in the M-step,
\begin{align}
\rijt &\leftarrow\  \frac{\Hrijt+\beta}{\alpha+2},\\
\lambdai &\leftarrow\  \frac{1}{J}\sum_{j}\frac{1}{\Trijn|\biH\ui|^2}\ui^{\Htenchi}\Hrijn\ui,\\
\Rin &\leftarrow\  \Rinilrma+\lambdai\bi\biH,\\
\rijn &\leftarrow\  \frac{1}{M}\tr\bigl((\Rin)^{-1}\Hrijn\bigr),
\end{align}
where $\Tthc=\{\Trijt,\Trijn,\Tlambdai\}$ is the set of up-to-date parameters and $\ui$ is the eigenvector of $\Rinilrma$ that corresponds to the zero eigenvalue.

\section{Proposed RCSCME}
\subsection{Motivation}
In conventional RCSCME, the direction of deficient basis $\bi$ is fixed and only its scale $\lambdai$ is estimated. However, the candidate of the deficient basis is not unique because any complex vector that is not included in the space spanned by column vectors of $\Rinilrma$ can be used. In this paper, we propose a new vector optimization algorithm that estimates the deficient basis $\bi$ itself, and derive new update rules based on the EM algorithm. This method can be interpreted as a generalization of conventional RCSCME.

\subsection{New formulation of full-rank SCM and derivation of update rules}
To simultaneously parameterize the direction and the scale of the deficient basis in $\Rinilrma$, we model the full-rank noise SCM $\Rin$ as follows:

\begin{align}
\label{eq:byC}
\Rin &= \Rinilrma + \ci\ciH,
\end{align}
where $\ci\in\mathbb{C}^M$ is the deficient basis vector in $\Rinilrma$.
We model all the other variables in the same way as conventional RCSCME.

We estimate the parameters $\thp=\{\rijt,\rijn,\ci\}$ by MAP estimation based on the EM algorithm with the latent variables $\sijt$ and $\cijn$.
A $Q$ function is defined by the expectation of the complete-data
log-likelihood w.r.t. $p(\sijt,\cijn|\xij;\Tthp)$ as
\begin{align}
Q(\thp;\Tthp)&=\sum_{i,j}\biggl[-(\alpha+2)\log\rijt-M\log\rijn-\log\det\Rin\biggr.\nonumber\\
 &\phantom{=}\biggl.\mbox{}-\frac{\Hrijt+\beta}{\rijt}-\frac{\tr\bigl((\Rin)^{-1}\Hrijn\bigr)}{\rijn}\biggr]+\mathrm{const.},
 \label{eq:Qfunc}
\end{align}
where $\Tthp=\{\Trijt,\Trijn,\Tci\}$ is the set of up-to-date parameters and $\mathrm{const.}$ are the constant terms that do not depend on $\thp$.

In the E-step, $\Hrijt$ and $\Hrijn$ are obtained in the same way as the conventional RCSCME as follows:
\begin{align}
\TRin &= \Rinilrma + \Tci\TciH,\\
\TRinjx &= \Trijt\ait(\ait)^{\Htenchi}+\Trijn\TRin,\\
\Hrijt &= \Trijt-(\Trijt)^2(\ait)^{\Htenchi}(\TRinjx)^{-1}\ait\nonumber\\
&\phantom{=}\mbox{}+|\Trijt\xij^{\Htenchi}(\TRinjx)^{-1}\ait|^2,\\
\Hrijn &= \Trijn\TRin-(\Trijn)^2\TRin(\TRinjx)^{-1}\TRin\nonumber\\
&\phantom{=}\mbox{}+(\Trijn)^2\TRin(\TRinjx)^{-1}\xij\xij^{\Htenchi}(\TRinjx)^{-1}\TRin.  
\end{align}

Among some update rules in the M-step, we especially describe the derivation of the update rule of $\ci$. First, we differentiate the $Q$ function with respect to $\Lci$ as follows:
\begin{align}
    \frac{\partial Q}{\partial \Lci}&=-J(\Rin)^{-1}\ci+J(\Rin)^{-1}\HS(\Rin)^{-1}\ci,
\end{align}
where $^*$ denotes the complex conjugate and $\HS$ is defined as follows:
\begin{align}
    \HS:=\frac{1}{J}\sum_{j}\frac{1}{\Trijn}\Hrijn.
\end{align}
From the equation $\partial Q/\partial \Lci=\bm{0}$, we obtain
\begin{align}
    \ci=\HS(\Rin)^{-1}\ci.
    \label{eq:cHRc}
\end{align}
It is diffult to directly solve (\ref{eq:cHRc}) w.r.t. $\ci$ since $\Rin$ itself includes $\ci$. Hence, we use the following claim to resolve the equation.
\begin{claim}
\label{claim:1}
Let $\ui\in\mathbb{C}^M$ be a vector that satisfies $\Rinilrma\ui=\bm{0}$. Then, the following holds:
\begin{align}
    (\Rin)^{-1}\ci&=\frac{\ui}{\ciH\ui}.
\end{align}
\end{claim}
\begin{proof}
First, it holds that
\begin{align}
    \Rin\ui&=\Rinilrma\ui+\ci\ciH\ui\nonumber\\&=(\ciH\ui)\ci.
\end{align}
Then, by multiplying $(\Rin)^{-1}/(\ciH\ui)$, we can obtain
\begin{align}
    (\Rin)^{-1}\ci&=\frac{\ui}{\ciH\ui}.
\end{align}
\end{proof}
By applying Claim \ref{claim:1} to (\ref{eq:cHRc}), the equation can be deformed as
\begin{align}
\label{eq:houkou}
    (\ciH\ui)\ci=\HS\ui.
\end{align}
By multiplying $\ui^{\Htenchi}$, we can obtain
\begin{align}
    |\ciH\ui|^{2}=\ui^{\Htenchi}\HS\ui,
\end{align}
and consequently $\ciH\ui$ can be represented as
\begin{align}
    \ciH\ui=\exp(-{\rm j}\phi)\sqrt{\ui^{\Htenchi}\HS\ui},
\end{align}
where $\phi\in[0, 2\pi)$ is an arbitrary constant
and ${\rm j}$ is an imaginary unit.
Hence, the solution of the equation  $\partial Q/\partial \Lci=\bm{0}$ is
\begin{align}
    \ci=\frac{\exp({\rm j}\phi)}{\sqrt{\ui^{\Htenchi}\HS\ui}}\HS\ui.
\end{align}
We use $\phi=0$ in our update rules since we model $\Rin$ by $\ci\ciH$, which does not depend on $\phi$. We obtain the update rules of all other parameters than $\ci$ in the same way as the conventional RCSCME. Finally, in the M-step, we update the parameters as follows:
\begin{align}
\rijt &\leftarrow\  \frac{\Hrijt+\beta}{\alpha+2},\\
\ci &\leftarrow\  \frac{\HS\ui}{\sqrt{\ui^{\Htenchi}\HS\ui}} ,\label{eq:honsitu}\\
\Rin &\leftarrow\  \Rinilrma+\ci\ciH,\\
\rijn &\leftarrow\  \frac{1}{M}\tr\bigl((\Rin)^{-1}\Hrijn\bigr).
\end{align}
(\ref{eq:honsitu}) can be interpreted as the {\it vector direction adaptation} by rotating $\ui$ via $\HS$ composed of posterior.

\section{Experimental Evaluation}
\drawfig{t}{0.80\linewidth}{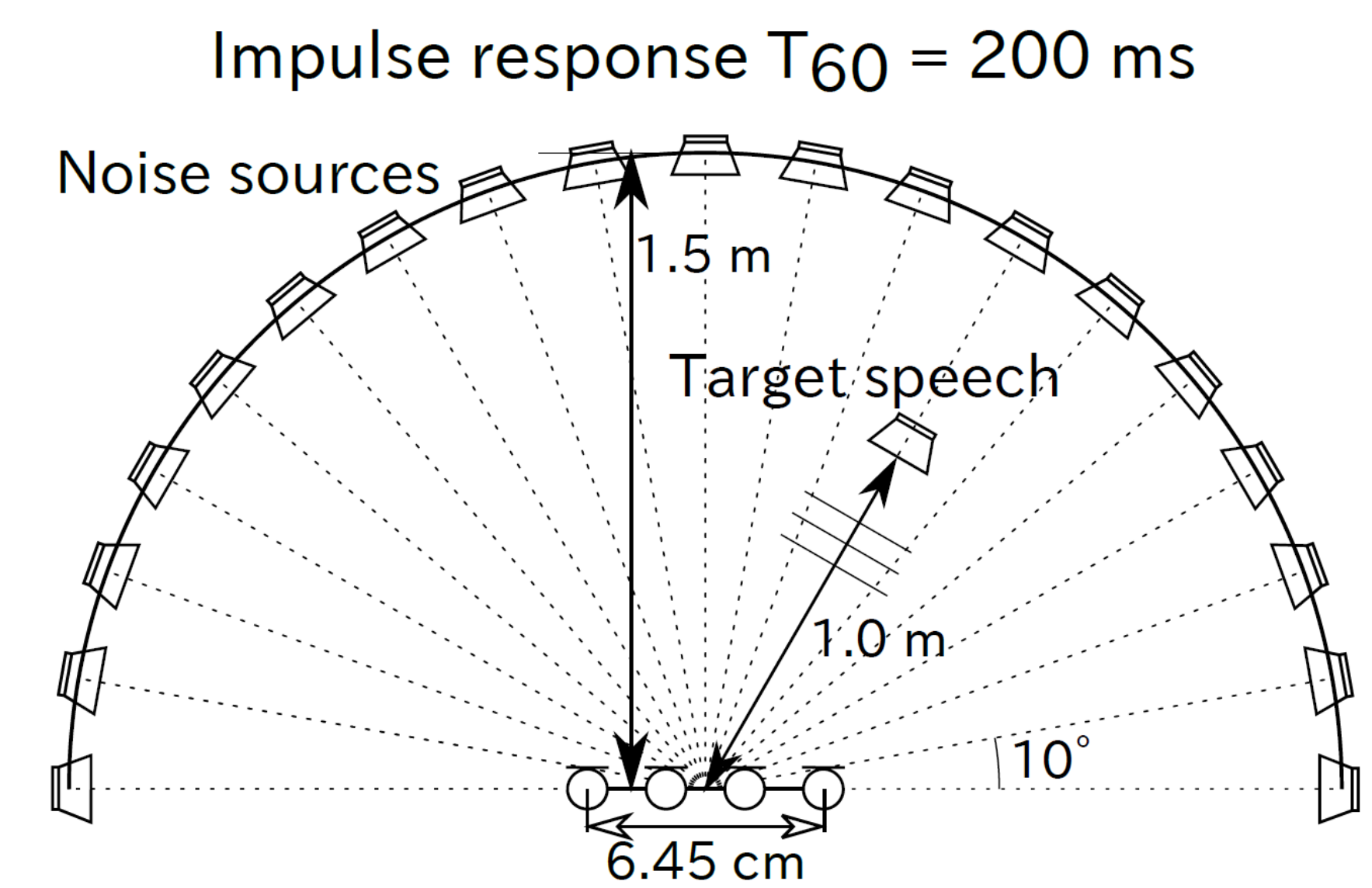}
        {\label{fig:impulse}Recording conditions of impulse responses.}
\drawfig{t}{0.7\linewidth}{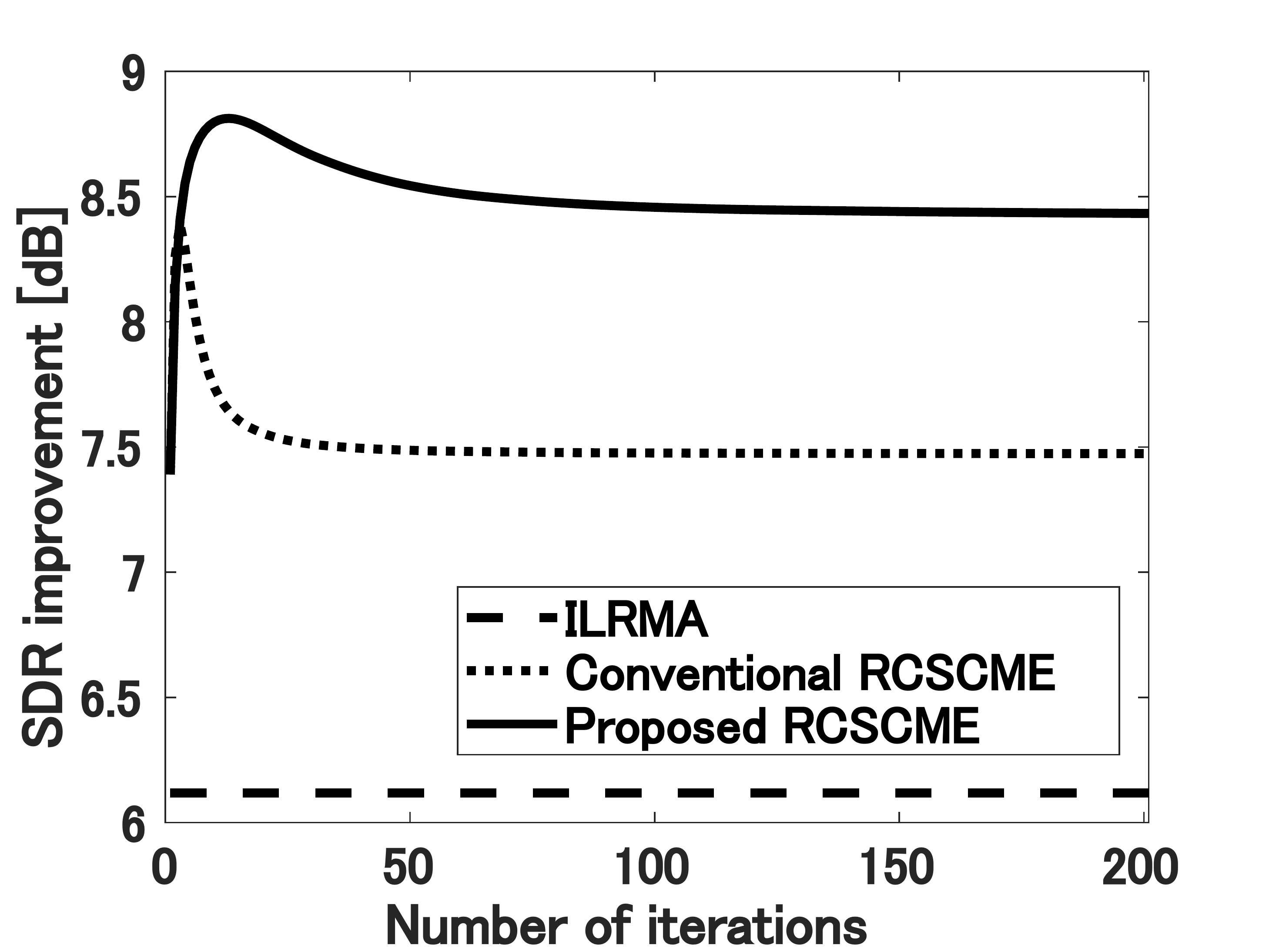}
        {\label{fig:babble}Behavior of SDR improvement in babble noise case.}

\subsection{Experimental conditions}
To confirm the efficacy of the proposed RCSCME, we conducted a BSE experiment using a simulated mixture of a target speech source and diffuse noise. To simulate the mixture, we convoluted dry sources with impulse responses from each position to four microphones as shown in Fig.~\ref{fig:impulse}. The diffuse noise was simulated by simultaneous reproduction from 19 positions and the target speech arrived from a closer position than each position of the diffuse noise. As the target speech source, we utilized six speech signals obtained from JNAS~\cite{KItou1999_JNAS}. We imitated the babble, station, traffic, and cafe noises.
We used 19 JNAS speech signals as the babble noise. The station, traffic, and cafe noise signals were obtained from DEMAND~\cite{Joachim2013_DEMAND}. An STFT was performed by using a 64-ms-long Hamming
window with a 32-ms-long shift. The speech-to-noise ratio was set
to 0 dB.

 We compared eight methods, namely, ILRMA~\cite{intro:ILRMA}, blind spatial subtraction array (BSSA)~\cite{2009zanryu}, multichannel Wiener filter with single-channel noise power estimation (MWF1)~\cite{MWF1}, multichannel Wiener filter with multichannel noise power estimation (MWF2)~\cite{MWF2}, the original FastMNMF~\cite{intro:fMNMF2}, FastMNMF initialized by ILRMA (ILRMA +FastMNMF), the conventional RCSCME~\cite{INTRO:KUBO} and the proposed RCSCME. As for BSSA, ILRMA was used instead of frequency-domain independent component
analysis in ~\cite{2009zanryu} and set the oversubtraction and flooring parameters to $1.4$ and $0$, respectively. For FastMNMF and ILRMA, the nonnegative matrix factorization (NMF) variables were initialized by nonnegative random
values, and the demixing matrix was initialized by the identity matrix. As for ILRMA+FastMNMF, the NMF variables were handed over from ILRMA to FastMNMF. Also, the SCM
was initialized by $\bm{a}_{i,n_{\rm t}}\bm{a}_{i,n_{\rm t}}^{\Htenchi}+\epsilon\II$
for ILRMA+MNMF and
$\bm{a}_{i,n_{\rm t}}\bm{a}_{i,n_{\rm t}}+\epsilon\sum_{n\not=n_{\rm t}}\bm{a}_{i,n}\bm{a}_{i,n}^{\Htenchi}$ for ILRMA+FastMNMF, where
$\II$ is the identity matrix and $\epsilon$ was set to $10^{-5}$. In ILRMA, which was used as the preprocessing for each method in the experiment, the number of bases was 10 and the number of iterations was 50. As for all methods, by selecting the channel whose kurtosis was the maximum in all channels, we blindly decided the index $n_{\rm t}$ of the target source in the demixed signals by ILRMA. In the conventional RCSCME, we utilized the minimum positive eigenvalue $\sigma_i$ of $\Rinilrma$ as the initial value of $\lambdai$. We used the inverse gamma distribution parameters $\alpha=2.5$ and $\beta=10^{-16}$ in the conventional RCSCME, which showed the best separation performance in the experiment of \cite{LAST:KUBO}. $\ci$ was initialized by $\sqrt{\sigma_i}\ui$ and we chose the parameters $\alpha=0.1$ and $\beta=10^{-16}$ in the proposed RCSCME experimentally. Source-to-distortion ratio (SDR) improvement \cite{EVincent2006_BSSEval} was used as a total evaluation
score. The SDR improvement was averaged over 10 parameter-initialization random seeds, four target directions, and six target speech sources; thus each SDR improvement score is the average value over 240 trials.

\begin{flushleft}
\begin{table}[tb]
\vspace{-4mm}
\caption{\label{table:babble}SDR improvements [dB] for each method and noise case.
Each term represents ”peak score / score after 200 iterations”}
{\arrayrulewidth=2pt
\setlength{\tabcolsep}{-2pt}
\begin{tabular}{c}
\hline
{
\setlength{\tabcolsep}{5pt}
\arrayrulewidth=0.4pt
  \begin{tabular}{@{\hspace{ 0mm}}c@{\hspace{ 0mm}}|cccc} 
    \begin{tabular}{l}
        Method
    \end{tabular} & 
    \begin{tabular}{c}
        babble\\noise
    \end{tabular} & 
    \begin{tabular}{c}
        station\\noise
    \end{tabular} &
    \begin{tabular}{c}
        traffic\\noise
    \end{tabular} &
    \begin{tabular}{c}
        cafe\\noise
    \end{tabular} \\\hline\hline
    \begin{tabular}{c}
        ILRMA
    \end{tabular} & 6.1 / - & 6.2 / - & 4.7 / - & 6.4 / - \\\hline
    \begin{tabular}{c}
        BSSA
    \end{tabular} & 6.8 / - & 6.9 / -   & 5.7 / - & 7.2 / - \\\hline
    \begin{tabular}{c}
        MWF1
    \end{tabular} & 6.1 / - & 6.9 / -  & 5.8 / - & 7.0 / - \\\hline
    \begin{tabular}{c}
        MWF2
    \end{tabular} & 6.9 / - & 7.2 / -   & 5.6 / - & 7.4 / - \\\hline
    \begin{tabular}{c}
        FastMNMF
    \end{tabular} & 1.7 / 1.7 & 2.6 / 2.5 & 2.9 / 2.8 & 2.6 / 2.6\\\hline
    \begin{tabular}{c}
        ILRMA\\+FastMNMF
    \end{tabular} & 6.5 / 6.2 & 6.6 / 6.6 & 5.4 / 5.4 & 7.3 / 7.3 \\\hline
    \begin{tabular}{c}
        Conventional\\RCSCME
    \end{tabular} & 8.4 / 7.5 & 9.5 / 8.9 & 7.5 / 7.2 & 9.6 / 8.9 \\\hline
    \begin{tabular}{c}
        Proposed\\RCSCME
    \end{tabular} & {\bf 8.8} / {\bf 8.4} & {\bf 10.7} / {\bf 10.4} & {\bf 8.8} / {\bf 8.7} & {\bf 10.6} / {\bf 10.2}  \\ 
  \end{tabular}
  }
  \\ \hline
  \end{tabular}
  }
  \vspace{-3mm}
\end{table}
\end{flushleft}
\vspace{-10mm}
\subsection{Result}
Figure \ref{fig:babble} shows the behavior example of the averaged SDR improvement of the proposed and the conventional RCSCME for each iteration under a babble noise condition. Although the preprocessing, i.e., ILRMA, is an iterative method, we show the averaged SDR improvement of ILRMA after 50 iterations as a reference. From this figure, we can see that the SDR improvement of the conventional and proposed RCSCMEs have a peak and the proposed RCSCME outperforms the conventional RCSCME from the viewpoint of both the peak score and the score after 200 iterations. 

Table \ref{table:babble} shows SDR improvements for each method under each noise condition. For FastMNMF, ILRMA+FastMNMF, conventional RCSCME, and proposed RCSCME, which are iterative methods, we display both the peak score and the score after 200 iterations. From this table, we reveal that the proposed RCSCME
outperforms all the conventional methods under all the noise conditions.

\section{Conclusion}
In this paper, we proposed a new algorithmic extension of RCSCME to improve the BSE performance. In the conventional RCSCME, a direction of the deficient basis is fixed and only the scale is estimated. In the proposed RCSCM, we accurately estimated the deficient basis itself as a vector variable by solving a vector optimization problem. Also, we derived a new update rules based on the EM algorithm. We confirmed that the proposed method outperformed conventional methods under several noise conditions.

\clearpage
\bibliographystyle{IEEEbib}
\bibliography{strings}

\end{document}